\RequirePackage{silence}
\documentclass[a4paper]{article}
\newcommand{\MINC}{\Minc}
\newcommand{\num}{\mathrm{\sharp}}
\newcommand{\PER}{\mathrm{PER}}
\newcommand{\SUPER}{\mathrm{S\text-PER}}

\newcommand{\BinTreeMincstrictSAT}{\ensuremath{\Minc^s\text-{\problemFont{tree}}\text-\SAT}\xspace}
\newcommand{\FinBinTreeMincstrictSAT}{\ensuremath{\Minc^s\text-{\problemFont{fintree}}\text-\SAT}\xspace}

\usepackage{hyperref}
\usepackage{minc}
\usepackage{verbatim}
\usepackage[numbers,sort]{natbib}
\makeatletter
\renewcommand\bibsection%
{
  \section*{\refname
    \@mkboth{\MakeUppercase{\refname}}{\MakeUppercase{\refname}}}
}
\makeatother

\title{Satisfiability of Modal Inclusion Logic:\\ Lax and Strict Semantics}

\newcommand{\inst}[1]{$^{#1}$}

\author{Lauri Hella\inst{1} \and Antti Kuusisto\inst{2} \and Arne Meier\inst{3}\and Heribert Vollmer\inst{3}}

\date{\small
\inst{1} Faculty of Natural Sciences,
University of Tampere,
Kanslerinrinne 1 B,\\
33014 University of Tampere, Finland,
\texttt{lauri.hella@uta.fi}\\
\inst{2}Fachbereich Mathematik und Informatik, Universität Bremen, Bibliothekstr.~1, 28359~Bremen,
Germany.
\texttt{antti.j.kuusisto@gmail.com} \\
\inst{3} Institut für Theoretische~Informatik,
Leibniz Universität Hannover,
Appelstr.~4, 30167~Hannover, Germany,
\texttt{\{meier,vollmer\}@thi.uni-hannover.de}}

\usepackage{hyperref}

\begin{document}

\maketitle

\begin{abstract}
We investigate the computational complexity of the satisfiability problem of modal inclusion logic. 
We distinguish two variants of the problem: one for the strict and another one for the lax semantics. 
Both problems turn out to be $\EXPTIME$-complete on general structures. 
Finally, we show how for a specific class of structures $\NEXPTIME$-completeness for these problems under strict semantics can be achieved.
\end{abstract}

\section{Introduction}  

Dependence logic was introduced by Jouko Väänänen \cite{vaananen07} in 2007. 
It is a first-order logic that enables one to talk about dependencies between variables explicitly.
It thereby generalizes Henkin quantifiers and also, in a sense, Hintikka's independence-friendly logic. 
Dependence logic can be used to formalize phenomena from a plethora of scientific disciplines such as database theory, social choice theory, cryptography, quantum physics, and others. 
It extends first-order logic by specific terms $\dep{x_{1},\dots,x_{n-1},x_{n}}$ known as dependence atoms, expressing that the value of the variable $x_{n}$ depends on the values of $x_{1},\dots,x_{n-1}$, i.e., $x_{n}$ is functionally determined by $x_{1},\dots,x_{n-1}$.
As such dependence does not make sense when talking about single assignments, formulas are evaluated over sets of assignments, called \emph{teams}. 
The semantics of the atom $\dep{x_{1},\dots,x_{n-1},x_{n}}$ is defined such that it is true in a team $T$ if in the set of all assignments in $T$, the value of $x_{n}$ is functionally determined by the values of $x_{1},\dots,x_{n-1}$.

In addition to dependence atoms, also generalised
dependency atoms have been introduced in the literature.
Examples include the independence atom (asserting that two sets of variables are informationally independent in a team), the non-emptiness atom (asserting that the team is non-empty), and, most importantly to the present paper, the inclusion atom $\vec{x}\subseteq\vec{y}$ for vectors of variables $\vec{x},\vec{y}$, asserting that in a team, the set of tuples assigned to $\vec{x}$ is included in the set of tuples assigned to $\vec{y}$. 
This corresponds to the definition of inclusion dependencies in database theory, which state that all tuples of values taken by the attributes $\vec{x}$ are also taken by the attributes $\vec{y}$. The notion of a generalized atom has been formally defined in \cite{Ku15}.

V\"a\"an\"anen \cite{vaananen08b} also introduced dependence atoms into modal logic. 
There teams are sets of worlds, and a dependence atom $\dep{p_{1},\dots,p_{n-1},p_{n}}$ holds in a team $T$ if there is a Boolean function that determines the value of $p_{n}$ in 
each world in $T$ from the values of $p_{1},\dots,p_{n-1}$.
The so-obtained modal dependence logic \MDL was studied from the point of view of expressivity and complexity in \cite{sev09}. 
Following the above mentioned developments in first-order dependence logic, modal dependence logic was also extended by generalized dependency atoms in \cite{KMSV14}, such as, e.g., independence atoms and inclusion atoms.

In the context of first-order dependence logic and its variants, two alternative kinds of team semantics have been distinguished, \emph{lax} and \emph{strict semantics} \cite{Galliani12}. 
Lax semantics is the standard team semantics, while strict semantics is obtained from lax semantics by introducing some additional uniqueness and strictness properties.
In the modal context, these additional constraints mainly concern the diamond modality $\Diamond$. 
In lax semantics, a formula $\Diamond\varphi$ holds in a team $T$ if there is a team $S$ such that every world in $T$ has at least one successor in $S$ and $\varphi$ holds in $S$. 
(Also, the worlds in $S$ are required to have a predecessor in $T$.) 
In strict semantics, we require that $S$ contains, for every world in $T$, a unique successor given by a  surjection $f:T\rightarrow S$.
(In first-order logic, strict semantics for the existential quantifier is defined similarly.) 
In both modal and first-order context, the operator known as \emph{splitjunction} (which corresponds to disjunction) is also defined differently for lax and strict
semantics (see  Section \ref{preliminaries} below).

For many variants of first-order and modal dependence logic, there is no distinction in expressive power between the two semantics. However, the choice of semantics plays a role in independence and inclusion logics, i.e., (first-order) logics with team
semantics and with independence or inclusion atoms.
For example, in the first-order case, inclusion logic with strict semantics has the same expressive power as dependence logic, i.e., \ESO (existential second-order logic) \cite{ghk13} and thus captures \NP, while with lax semantics, inclusion logic is equivalent to greatest fixpoint logic and consequently can express exactly the polynomial-time decidable properties of finite ordered structures.

%
%

The purpose of the present paper is to investigate the complexitly of the satisfiability problem of
modal inclusion logic; we cover both the case of lax as well as strict semantics.
We show that in both cases, the problem is $\EXPTIME$-complete. Furthermore,
the same results hold already for propositional inclusion logic, meaning roughly that the
lower bounds of
our results can be obtained even without including modal operators into the picture.
The $\EXPTIME$-completeness result for lax semantics is
obtained by identifying the upper bound via a translation to standard multimodal logic with the
global modality and converse modalities, and the lower bound is established via a
reduction from a certain succinct encoding of a $\P$-complete problem introduced in this paper.
The case for strict semantics is similar but requires some reasonably 
straightforward yet interesting modifications to
the arguments for lax semantics.
All the complexity results identified here hold also for finite satisfiability.

The conference version of this paper \cite{hkmv15} claims different complexities regarding the satisfiability problem with respect to the two underlying semantics: the satisfiability problem of
modal inclusion logic is erroneously claimed $\NEXPTIME$-complete there. In this paper we
fix this issue by providing detailed proofs for all the cases discussed. We also identify a case
where the ideas of the faulty argument from \cite{hkmv15} actually go through by investigating
modal inclusion logic in restriciton to pointed binary trees, i.e.,
binary trees such that the initial team is the singleton
containing the root only. We show that the satisfiability
problem of modal inclusion logic under strict semantics is $\NEXPTIME$-complete
over pointed binary trees.

The paper is organized as follows. After the preliminaries in Section \ref{preliminaries}, we
investigate the satisfiability problem of modal inclusion logic under lax semantics in
Section \ref{sect:coco}. The upper
bound is discussed in \ref{laxup} and the lower bound in \ref{laxlow}.
The corresponding analysis for strict semantics is then given in Section \ref{strictall}.
In Section \ref{trees} we consider strict semantics in restriction to pointed binary trees and
conlude in Section \ref{conclusion}.

\section{Preliminaries}\label{preliminaries}

Let $\Pi$ be a countably infinite set of proposition symbols.
The set of formulas of \emph{modal inclusion logic} $\Minc$ is defined inductively by the following grammar.
$$
\varphi \ddfn
			p\mid 
			\lnot p\mid 
			(\varphi_1\land\varphi_2)\mid 
			(\varphi_1\lor\varphi_2)\mid
p_1\cdots p_k\subseteq q_1\cdots q_k\mid
			\Box\varphi\mid
			\Diamond\varphi,
$$
where $p,p_1,\dots,p_k,q_1,\dots,q_k\in\Pi$ are proposition symbols and $k$
is any positive integer. The formulas $p_1\cdots p_k\subseteq q_1\cdots q_k$
are called \emph{inclusion atoms}.
For a set $\Phi\subseteq\Pi$, we let $\Minc(\Phi)$ be 
the sublanguage where only propositions from $\Phi$ are used.
Observe that formulas are essentially in negation normal form; negations may
occur only in front of proposition symbols.
A Kripke model is a structure 
$M=(W,R,V)$, where $W\not=\emptyset$
is a set (the domain of the model, or the set of worlds/states),
$R\subseteq W\times W$ is a binary relation (the accessibility or transition relation),
and $V\colon\Pi\rightarrow\mathcal{P}(W)$ is a
\emph{valuation} interpreting the proposition symbols.
Here $\mathcal{P}$ denotes the power set operator.
The language of basic unimodal logic is the sublanguage of $\Minc$
without formulas $p_1\cdots p_k\subseteq q_1\cdots q_k$.
We assume that the reader is familiar with standard 
Kripke semantics of modal logic; we 
let $M,w\Vdash\varphi$ denote the assertion
that the point $w\in W$ of the model $M$ satisfies $\varphi$
according to standard Kripke semantics.
We use the symbol $\Vdash$ in order to refer to
satisfaction according to standard Kripke semantics,
while the symbol $\models$ will be reserved for
\emph{team semantics}, to be defined below,
which is the semantics $\Minc$ is based on.
Let $T$ be a subset of 
the domain $W$ of a Kripke model $M$.
The set $T$ is called a \emph{team}.
The semantics of the inclusion atoms $p_1\cdots p_k\subseteq q_1\cdots q_k$ is
defined such that $M,T\models p_1\cdots p_k\subseteq q_1\cdots q_k$
if and only if for each $u\in T$, there exists a point $v\in T$ 
such that $$\bigwedge\limits_{i\, \in\, \{1,...,k\}}\bigl( u\in V(p_i)\Leftrightarrow v\in V(q_i)\bigr).$$
The intuition here is that every vector of truth values taken by $p_1,\dots,p_k$, is
included in the set of vectors of truth values taken by $q_1,\dots,q_k$.
Let $M = (W,R,V)$ be a Kripke model and $T\subseteq W$ a team. Define the set of successors of $T\subseteq W$ to be $R(T)\dfn\{s\in W\mid\exists s'\in T:(s',s)\in R\}$.
Also define $$R\langle T\rangle\dfn\{\, T\, '\subseteq W\mid\forall s\in T\exists s'\in T\, '\text{ s.t.\ }(s,s')\in R\text{ and }\forall s'\in T\, '\, \exists s\in T\text{ s.t.\ }(s,s')\in R\,\},$$ which we call the set of
allowed successor teams of $T$.
The following clauses, together with the above clause for inclusion atoms,
define \emph{lax semantics} (or lax team semantics)
for $\Minc$.
\[
\begin{array}{@{}l@{}l@{}l}
M,T\modelslax p\ &\Leftrightarrow &\  w\in V(p)\text{ holds for all }w\in T.\\
M,T\modelslax\neg p\ &\Leftrightarrow &\  w\not\in V(p)\text{ holds for all }w\in T.\\
M,T\modelslax\varphi\wedge\psi\ &\Leftrightarrow &\
M,T\modelslax\varphi\text{ and }M,T\modelslax\psi.\\
M,T\modelslax\varphi\vee\psi\ &\Leftrightarrow &\
M,S\modelslax\varphi\text{ and }M,S'\modelslax\psi
\text{ for some }S,S'\subseteq T\text{ such that }\\
& &\text{ we have }S\cup S' = T.\\ 
M,T\modelslax\Box\varphi\ &\Leftrightarrow &\ M,R(T)\modelslax\varphi.\\
M,T\modelslax\Diamond\varphi\ &\Leftrightarrow &\ \exists T\, '\in R\langle T\rangle:M,T\, '\modelslax\varphi.
%
\end{array}
\]
The other semantics for $\Minc$, \emph{strict semantics},
differs from the lax semantics only in its treatment of
the disjunction $\vee$ and diamond $\Diamond$. Thus, all other clauses in the
definition of $\modelsstrict$
are the same as those for $\modelslax$, and
the clauses for $\vee$ and $\Diamond$ are as follows.
\[
\begin{array}{lll}
M,T\modelsstrict\varphi\vee\psi\ &\Leftrightarrow &
M,S\modelsstrict\varphi\text{ and }M,S'\modelsstrict\psi
\text{ for some }S,S'\subseteq T\text{ such that }\\
& &S\cup S' = T\text{ and }S\cap S' = \emptyset.\\
M,T\modelsstrict\Diamond\varphi\ &\Leftrightarrow &$\text{ }$
M,f(T)\modelsstrict\varphi\text{ for some function }
f\colon T\rightarrow W
\text{ such that }\\
& &(u,f(u))\in R\text{ for all }u\in T.\
(\text{Here }f(T) = \{\, f(u)\, |\, u\in T\, \}.)
\end{array}
\]

Intuitively, the difference between the lax and strict semantics is as the terms suggest. In strict semantics, the division of a team with the \emph{splitjunction} $\lor$ is strict; no point is allowed to occur in both parts of the division contrarily to lax semantics. For $\Diamond$, strictness is
related to the  use of functions when finding a team of successors.
The difference between lax and 
strict semantics in first-order
inclusion logic \cite{Galliani12} is similar.

It is well known and easy to show that any formula of modal logic,
i.e., a formula of $\Minc$ \emph{without} inclusion atoms, is satisfied
by a team if and only if it is satisfied by every point in the team.

\begin{proposition}\label{flatness}
Let $\varphi$ be a formula of $\Minc$ \emph{without} inclusion atoms, and
let $T$ be a team on a Kripke model $M$. Then 
$$
	M,T\modelslax\varphi\iff M,T\modelsstrict\varphi\iff 
	\forall w\in T(M,w\Vdash\varphi).
$$
Here $\Vdash$ denotes
satisfaction in the standard sense of Kripke semantics. 
\end{proposition}

The equivalence in Proposition \ref{flatness} is the so-called \emph{flatness property}. It shows that
team semantics is essentially just a generalization of the classical (Kripke) semantics.
The satisfiability problem of $\Minc$ with lax (strict) semantics is the problem that asks,
given a formula $\varphi$ of $\Minc$, whether there exists a \emph{nonempty} team $T$
and a model $M$ such that $M,T\modelslax\varphi$ ($M,T\modelsstrict\varphi$) holds.
Note that the requirement of $T$ being nonempty is necessary: by the well-known \emph{empty team property}, $M,\emptyset\modelslax\varphi$ (and 
$M,\emptyset\modelsstrict\varphi$) holds for any formula $\varphi\in\Minc$.
Two different problems arise, depending on whether lax or strict semantics is used.
The corresponding finite satisfiability problems require that the satisfying models have a finite domain.

\section{Complexity of Satisfiability for Lax Semantics}
\label{sect:coco}

\subsection{Upper bound for lax semantics}\label{laxup}

In this section we show that the satisfiability and finite satisfiability problems
of \Minc with  lax semantics are in $\EXPTIME$.
The result is established by an equivalence preserving
translation to \emph{propositional dynamic logic} extended with the
global and converse modalities. It is well-known that this logic is
complete for $\EXPTIME$ (see \cite{blackburn, hema, eijck}).
In fact, we will only need
multimodal logic with the global modality and converse modalities
for our purposes.
Let $\Pi$ and $\mathcal{R}$
be countably infinite sets of
proposition symbols and binary relation symbols, respectively.
The following grammar defines a modal language $\mathcal{L}$.
$$\varphi\, \ddfn\, p\ |\ \neg\varphi\ |\ (\varphi_1\wedge\varphi_2)\ |\ \langle R\rangle \varphi\ |\ \langle R^{-1}\rangle\varphi\ |\ \langle E\rangle\varphi$$
Here $p\in\Pi$, $R\in\mathcal{R}$, and $E$ is a
novel symbol.
The (classical Kripke-style) semantics of $\mathcal{L}$ is defined with
respect to ordinary pointed Kripke models $(M,w)$ for multimodal logic.
Let $M = (W, \{R\}_{R\in\mathcal{R}}, V)$ be a Kripke model,
where $V\colon\Pi\rightarrow\mathcal{P}(W)$
is the \emph{valuation} function interpreting proposition symbols.
Let $w\in W$. The following clauses define the semantics of $\mathcal{L}$;
note that we use the turnstile $\Vdash$ instead of $\models$,
which is reserved for team semantics in this paper.
\[
\begin{array}{lll}
M,w\Vdash p & \Leftrightarrow\ & w\in V(p)\\
M,w\Vdash \neg\varphi &\Leftrightarrow& M,w\not\Vdash\varphi\\
M,w\Vdash \varphi_1\wedge\varphi_2 & \Leftrightarrow\ & M,w\Vdash\varphi_1\text{ and }M,w\Vdash\varphi_2\\
M,w\Vdash \langle R\rangle\varphi & \Leftrightarrow & M,u\Vdash\varphi \text{ for some }u\text{ such that }wRu\\
M,w\Vdash \langle R^{-1}\rangle\varphi & \Leftrightarrow & M,u\Vdash\varphi
\text{ for some }u\text{ such that }uRw\\
M,w\Vdash \langle E \rangle\varphi & \Leftrightarrow & M,u\Vdash\varphi
\text{ for some }u\in W\\
\end{array}
\]
%
%
%
%
%
%

%
We next define a satisfiability preserving translation from
modal inclusion logic into $\mathcal{L}$. We let $[R]$
and $[E]$ denote $\neg\langle R\rangle\neg$ and $\neg\langle E\rangle\neg$,
respectively. Before we fix the translation, we define some auxiliary formulas.
Let $\theta$ be a formula of \Minc.
We let $\mathit{SUB}(\theta)$ denote the set of subformulas of $\theta$;
we distinguish all instances of subformulas, so for example
$p\wedge p$ has \emph{three} subformulas 
(the right and the left instances of $p$ and the conjunction itself).
For each formula $\varphi\in\mathit{SUB}(\theta)$,
fix a fresh proposition symbol $p_{\varphi}$ that does not occur in $\theta$.
We next define, for each $\varphi\in\mathit{SUB}(\theta)$, a
novel auxiliary formula $\chi_{\varphi}$.
If $\varphi\in\mathit{SUB}(\theta)$ is a literal $p$ or $\neg p$, we define $\chi_{\varphi}\ \dfn\ [E]\bigl(\ p_{\varphi}\ \rightarrow\ \varphi\ \bigr).$
Now fix a symbol $R\in\mathcal{R}$, which will ultimately correspond
to the diamond used in modal inclusion logic.
For the remaining subformulas $\varphi$ of $\theta$, with
the exception of inclusion atoms, the formula $\chi_{\varphi}$ is defined as follows.
\begin{enumerate}
\item
$\chi_{\varphi\wedge\psi}\ \dfn\ 
[E]\bigl(\ (p_{\varphi\wedge\psi}\ \leftrightarrow\ 
p_\varphi)\ \wedge\ (p_{\varphi\wedge\psi} \leftrightarrow p_\psi)\ \bigr)$
\item
$\chi_{\varphi\vee\psi}\ \dfn\ 
[E]\bigl(\ p_{\varphi\vee\psi}\ \leftrightarrow\ (p_\varphi\vee p_\psi)\ \bigr)$
\item
$\chi_{\Box\varphi}\ \dfn\ 
[E]\bigl(\ (p_{\Box\varphi}\  \rightarrow\ [R] p_{\varphi})\ \wedge\ 
(p_{\varphi}
\  \rightarrow\ 
\langle R^{-1}\rangle p_{\Box\varphi})\ \bigr)$
\item
$\chi_{\Diamond\varphi}\ \dfn\ 
[E]\bigl(\ (p_{\Diamond\varphi}\  \rightarrow\ \langle R\rangle  p_{\varphi})
\wedge (p_{\varphi}
\  \rightarrow\ 
\langle R^{-1}\rangle p_{\Diamond\varphi}\ \bigr)\ \bigr)$
\end{enumerate}
We then define the formulas $\chi_{\alpha}$ where
$\alpha\in\mathit{SUB}(\theta)$ is an inclusion atom.
We appoint a fresh binary relation $R_{\alpha}$ for each inclusion
atom in $\theta$.
%
%
Assume $\alpha$ denotes the inclusion atom $p_1\cdots p_k\subseteq q_1\cdots q_k$.
We define
\[
\begin{array}{ll}
\chi_{\alpha}^+\ &\dfn\ \smallskip
\bigwedge\limits_{i\, \in\, \{1,...,k\}}
[E]\bigl(\, (p_{\alpha}\wedge p_i)\,\rightarrow\,
\langle R_{\alpha}\rangle( p_{\alpha}\wedge q_i)\, \bigr),\\ 
\chi_{\alpha}^-\ &\dfn\ \smallskip
\bigwedge\limits_{i\, \in\, \{1,...,k\}}
[E]\bigl(\, (p_{\alpha}\wedge \neg p_i)\,\rightarrow\,
\langle R_{\alpha}\rangle( p_{\alpha}\wedge \neg q_i)\, \bigr),\\ 
\chi_{\alpha}\, &\dfn\, \chi_{\alpha}^+\wedge \chi_{\alpha}^-\ \wedge\,
\bigwedge\limits_{i\, \in\, \{1,\dots,k\}}[E]\bigl(\, \langle R_{\alpha}\rangle q_i
\rightarrow [R_{\alpha}] q_i\, \bigr).
\end{array}
\]
Finally, we define $\text{ }$
$\varphi_{\theta}\ \dfn\ p_{\theta}\ 
\wedge \bigwedge\limits_{\varphi\, \in\, \mathit{SUB}(\theta)}\chi_{\varphi}\, .$

Note that clearly the size of the formula $\varphi_\theta$ is polynomial
with respect to the size of $\theta$.
\begin{theorem}\label{thm-sat-fsat-lax-exptime}
The satisfiability and finite satisfiability problems for
modal inclusion logic with lax semantics are in \EXPTIME.
\end{theorem}
\begin{proof}
We will show that any formula $\theta$ of modal inclusion
logic is satisfiable if and only if its translation $\varphi_{\theta}$ is.
Furthermore, $\theta$ is satisfiable over a domain $W$ if and only if
$\varphi_{\theta}$ is satisfiable over $W$, and therefore 
we also get the desired result for finite satisfiability;
$\mathcal{L}$ has the finite model property since it
clearly translates to two-variable logic via a simple
extension of the \emph{standard translation} (see
\cite{blackburn} for the definition of standard translation).
Let $M = (W,R,V)$ be a Kripke model.
Let $I(\theta)\subseteq\mathit{SUB}(\theta)$
be the set of inclusion atoms in $\theta$.
Assume that $M,X\modelslax\theta$, where $X$ is a nonempty team.
We next define a multimodal Kripke model
$N\, \dfn\, (W,R,\{R_{\alpha}\}_{\alpha\, \in\,
\mathit{I}(\theta)}, V\cup U),$
where $U\colon\{\, p_{\varphi}\ |\ \varphi\in\mathit{SUB}{(\theta)}\}\rightarrow\mathcal{P}(W)$
extends the valuation function $V$.
Define $U(p_{\theta}) = X$. Thus we have $M,U(p_{\theta})\modelslax\theta$.
Working from the
root towards the leaves of the parse tree of $\theta$,
we next interpret the remaining predicates $p_{\varphi}$ inductively
such that the condition $M, U(p_{\varphi})\modelslax\varphi$ is maintained.
Assume $U(p_{\psi\wedge\psi'})$ has been defined.
We define $U(p_{\psi}) = U(p_{\psi'}) = U(p_{\psi\wedge\psi'})$.
As $M,U(p_{\psi\wedge\psi'})\modelslax\psi\wedge\psi'$,
we have $M,U(p_{\psi})\modelslax\psi$ and $M,U(p_{\psi'})\modelslax\psi'$.
Assume then that $U(p_{\psi\vee\psi'})$ has been defined.
Therefore there exist sets $S$ and $S'$ such that 
$M,S\modelslax\psi$ and $M,S'\modelslax\psi'$, and furthermore, 
$S\cup S' = U(p_{\psi\vee\psi'})$.
We define $U(p_{\psi}) =  S$ and $U(p_{\psi'}) =  S'$.
Consider then the case where $U(p_{\Diamond\varphi})$ has been defined.
Call $T\dfn U(p_{\Diamond\varphi})$.
As $M,T\modelslax\Diamond\varphi$, there exists a set $T\hspace{0.4mm} '\subseteq W$
such that each point in $T$ has an $R$-successor in $T\hspace{0.4mm}'$,
and each  point in $T\hspace{0.4mm}'$ has an $R$-predecessor in $T$,
and furthermore, $M,T\hspace{0.4mm}'\modelslax\varphi$.
We set $U(p_{\varphi}) \dfn T\hspace{0.4mm}'$.
Finally, in the case for $p_{\Box\varphi}$,
the set $U(p_{\varphi})$ is defined to be the set of points
that have an $R$-predecessor in $U(p_{\Box{\varphi}})$.
We have now fixed an interpretation for each of the
predicates $p_{\varphi}$. The relations $R_{\alpha}$, where $\alpha$ is an
inclusion atom, remain to be interpreted.
Let $p_1\cdots p_k\subseteq q_1\cdots q_k$
be an inclusion atom in $\theta$, and denote this
atom by $\alpha$.
Call $T \dfn U(p_{\alpha})$.
Let $u\in T$. Since $M,T\modelslax\alpha$, there
exists a point $v\in T$ such that for each $i\in\{1,\dots,k\}$,
$u\in V(p_i)$ if and only if $v\in V(q_i)$. Define the pair $(u,v)$ to
be in $R_{\alpha}$. In this fashion, consider each point $u$ in $T$ 
and find exactly one corresponding point $v$ for $u$, and put the pair $(u,v)$
into $R_{\alpha}$. This fixes the interpretation of $R_{\alpha}$.
Let $w\in X = U(p_{\theta})$. 
Recalling how the sets $U(p_{\varphi})$ were defined, it is now
routine to check that $N,w\Vdash\varphi_{\theta}$.
We then consider the converse implication of the
current theorem. Therefore we assume that $N,w\Vdash\varphi_{\theta}$,
where $N$ is some multimodal Kripke
model in the signature of $\varphi_{\theta}$ and $w$ a
point in the domain of $N$.
We let $W$ denote the domain and $V$ the
valuation function of $N$.
For each $\varphi\in\mathit{SUB}(\theta)$, define the
team $X_{\varphi} \dfn V(p_{\varphi})$.
We will show by induction on the structure of $\theta$
that for each $\varphi\in\mathit{SUB}(\theta)$,
we have $N,X_{\varphi}\modelslax\varphi$.
Once this is done, it is clear that $M,X_{\theta}\modelslax\theta$,
where $M$ is the restriction of $N$ to the signature of $\theta$.
Furthermore, we
have $X_{\theta}\not=\emptyset$ as $w\in V(p_{\theta})$ (because $N,w\Vdash\varphi_{\theta}$).
Now recall the definition of the formulas $\chi_{\varphi}$,
where $\varphi\in\mathit{SUB}(\theta)$.
Let $p\in\mathit{SUB}(\theta)$.
It is clear that $N,X_p\modelslax p$, since
$N,w\Vdash \chi_p$. Similarly,
we infer that $N,X_{\neg q}\modelslax\neg q$ for $\neg q\in\mathit{SUB}(\theta)$.
Consider then a subformula $p_1\cdots p_k\subseteq q_1\cdots q_k$
of $\varphi$. Denote this inclusion atom by $\alpha$.
Consider a point $u\in X_{\alpha}$. If $u$ satisfies $p_i$
for some $i\in\{1,\dots,k\}$,
then we infer that since $N,w\Vdash\chi_{\alpha}^+$,
there exists a point $v_i\in X_{\alpha}$
that satisfies $q_i$.
Similarly, if $u$ satisfies $\neg p_j$,
we infer that since $N,w\Vdash\chi_{\alpha}^-$,
there exists a point $v_j\in X_{\alpha}$ that satisfies $\neg q_j$.
To conclude that $N,X_{\alpha}\modelslax \alpha$, it suffices
to show that all such points $v_i$ and $v_j$ can be chosen
such that $v_i = v_j$ for all $i,j\in\{1,\dots,k\}$.
This follows due to the third conjunct of $\chi_{\alpha}$.
Having established the basis of the induction,
the rest of the argument is straightforward.
We consider explicitly only the case where 
the subformula under consideration is $\Diamond\varphi$.
Here we simply need to argue that for each $u\in X_{\Diamond\varphi}$,
there exists a point $v\in X_{\varphi}$ such that $uRv$,
and for each $u'\in X_{\varphi}$, there exists a point $v'\in X_{\Diamond\varphi}$
such that $v'Ru'$. This follows directly, since $N,w\Vdash\chi_{\Diamond\varphi}$.\qed
\end{proof}

\subsection{Lower bound for lax semantics}\label{laxlow}

In this section we prove the satisfiability problem of $\MINC$ with lax semantics, $\MinclaxSAT$, 
to be hard for $\EXPTIME$. We do this by reducing to it the succinct version of the following
$\P$-hard problem which is closely related to the problem PATH SYSTEMS \cite[p. 171]{greenlaw}.

\begin{definition}
Let $\PER$ be the following problem: An instance of $\PER$ is a structure $\mA=(A,S)$ with $A=\{1,\ldots,n\}$ and $S\subseteq A^3$. A subset $P$ of $A$ is \emph{$S$-persistent} if it satisfies the condition 
\begin{itemize}
\item[($*$)] if $i\in P$, then there are $j,k\in P$ such that  $(i,j,k)\in S$.
\end{itemize}
$\mA$ is a positive instance if $n\in P$ for some $S$-persistent set $P\subseteq A$.
\end{definition}


It is well known that structures $(A,S)$ as above can be represented in a succinct form 
by using Boolean circuits.
Namely if $C$ is Boolean circuit with $3\cdot l$ input gates 
then it defines a structure $\mA_C=(A_C,S_C)$ given below.
We use here the notation $\sharp(a_1,\ldots,a_l)$ for the natural number 
$i$, whose binary representation is $(a_1,\ldots,a_l)$.
Let $A_C=\{1,\dots,2^l\}$, and for all $i,j,k\in A$, let $(i,j,k)\in S_C$ if and only if $C$ accepts the input tuple $$(a_1,\ldots,a_l,b_1,\ldots,b_l,c_1,\ldots,c_l)\in\{0,1\}^{3l},$$ where $i=\num(a_1,\ldots,a_l)$, $j=\num(b_1,\ldots,b_l)$ and $k=\num(c_1,\ldots,c_l)$.
We say that $C$ is a \emph{succinct representation} of 
$\mA_C$.

\begin{definition}
The succinct version of $\PER$, $\SUPER$, is the following problem: An instance of $\SUPER$ is a circuit $C$ with $3l$ input gates. $C$ is a positive instance, if $\mA_C$ is a positive instance of $\PER$.
\end{definition}

\begin{lemma}\label{circval}
$\SUPER$ is $\EXPTIME$-hard with respect to $\leqpm$-reductions.
\end{lemma}
\begin{proof}
\newcommand{\APSPACE}{\classFont{APSPACE}}

\newcommand{\mB}{\mathfrak{B}}
\newcommand{\acc}{\mathit{Acc}}
\newcommand{\rej}{\mathit{Rej}}
\newcommand{\Left}{\mathit{left}}
\newcommand{\Right}{\mathit{right}}
\newcommand{\AC}{\mathrm{AC}}
\newcommand{\RC}{\mathrm{RC}}
\newcommand{\RK}{\mathrm{rk}}
\newcommand{\depth}{\mathrm{TS}}

Let $M=(\Sigma,Q,\gamma,s_0,\delta)$ be an alternating 
Turing machine. That being so, $\Sigma$ is a finite tape alphabet, $Q$ is a finite set of states, 
the function $\gamma: Q\to\{\forall,\exists,\acc,\rej\}$ divides $Q$ according to the 
type of the states (universal, existential, accepting, rejecting), $s_0\in Q$ is the
initial state, and 
$\delta: \Sigma\times Q\to \mathcal{P}(\Sigma\times Q\times\{\Left,\Right,0\})$
is a transition function. 

Configurations of  $M$ are defined as usual. If $\alpha$ is a configuration, we write
$s^\alpha$ for its state. Furthermore, we write
$\alpha\mapsto_M \beta$ if $\alpha$ and $\beta$ are configurations such that
$\beta$ can be obtained from $\alpha$ by a transition allowed by $\delta$. 
Without loss of generality we assume that $\Sigma=\{0,1\}$, and 
$|\delta(0,s)|=|\delta(1,s)|=2$ for all $s$ such that $\gamma(s)\in\{\forall,\exists\}$.
On that account, if $\gamma(s^\alpha)\in\{\forall,\exists\}$ for a configuration $\alpha$,
then there are exactly two configurations $\beta$ such that $\alpha\mapsto_M \beta$.
On the other hand, if $\gamma(s)\in\{\acc,\rej\}$, we assume that 
$\delta(0,s)=\delta(1,s)=\emptyset$. As a result, the computation halts in a configuration
$\alpha$ such that $\gamma(s^\alpha)\in\{\acc,\rej\}$. 

The sets $\AC(M)$ of accepting configurations and $\RC(M)$ of rejecting 
configurations of $M$ are defined recursively in the usual way:
\begin{itemize}
\item If $\gamma(s^\alpha)=\acc$, then $\alpha\in \AC(M)$.
\item If $\gamma(s^\alpha)=\forall$ and $\beta\in\AC(M)$ for all
$\beta$ such that $\alpha\mapsto_M\beta$, then $\alpha\in\AC(M)$. 
\item If $\gamma(s^\alpha)=\exists$ and there is $\beta\in\AC(M)$ 
such that $\alpha\mapsto_M\beta$, then $\alpha\in\AC(M)$.
\item If $\gamma(s^\alpha)=\rej$, then $\alpha\in \RC(M)$.
\item If $\gamma(s^\alpha)=\forall$ and there is $\beta\in\RC(M)$ 
such that $\alpha\mapsto_M\beta$, then $\alpha\in\RC(M)$. 
\item If $\gamma(s^\alpha)=\exists$ and $\beta\in\RC(M)$ 
for all $\beta$ such that $\alpha\mapsto_M\beta$, then $\alpha\in\RC(M)$.
\end{itemize}


The machine $M$ accepts (rejects) a word $w\in \Sigma^*$ if $\alpha_w\in\AC(M)$
($\alpha_w\in\RC(M)$, respectively)
for the initial configuration $\alpha_w$ of $M$ with $w$ as input.
We denote the language $\{w\in\Sigma^*\mid M\text{ accepts }w\}$
by $L_M$. The machine $M$ decides the language $L_M$, if in addition
$M$ rejects all inputs $w\not\in L_M$. 

The class $\APSPACE$ consists of all languages $L_M$, where $M$ is an alternating
Turing machine $M$ that uses only polynomial number of tape cells.
It is well known that if $L\in\APSPACE$, then there is a polynomial space 
alternating machine $M$ that decides $L$, and which is acyclic in the sense that
there are no $\mapsto_M$-cycles among the configurations of $M$.

Now turn to the proof of the lemma. Let $L\in\APSPACE$, and let $M$ be an alternating Turing machine that works in polynomial space such that $L=L_M$. 
For each input word $w\in\{0,1\}^*$ we construct a circuit
$C_{M,w}$ in polynomial time from $w$ such that $C_{M,w}$ is a 
positive instance of $\SUPER$ if and only if $M$ accepts $w$.
This shows that $L_M$ is reducible to $\SUPER$, and since this holds 
for every language $L_M$ in $\APSPACE$, and $\APSPACE=\EXPTIME$,
it follows that $\SUPER$ is $\EXPTIME$-hard.

As explained above, we may assume that $\mapsto_M$ is acyclic, and
$M$ decides the language $L$. Let  $f$ be the polynomial such that for 
all inputs of length $n$, $M$ uses at most $f(n)$ tape cells.
Accordingly, if $w=w_1\ldots w_n\in\{0,1\}^n$ is an input word for $M$,
then we can encode the possible configurations of $M$ during
the computation on input $w$ with tuples
$$
	(a_1,\ldots,a_{2m+k})\in \{0,1\}^{2m+k},
$$
where $m:=f(n)$, as follows:
\begin{itemize}
\item $(a_1,\ldots,a_m)$ represents the contents of the tape in $\alpha$,
\item $(a_{m+1},\ldots,a_{2m})$ encodes the position of the tape head in $\alpha$:
$a_{m+i}=1$ if and only if the head is on the $i$-th cell,
\item $(a_{2m+1},\ldots,a_{2m+k})$ encodes the state $s^\alpha$: $a_{2m+(i+1)}=1$ if and only if 
$s^\alpha=s_i$, where $s_0,\ldots,s_{k-1}$, lists $Q$ in some fixed order.
\end{itemize}

The circuit $C_{M,w}$ will now be defined in such a way that the following
conditions hold:
\begin{enumerate}
\item $C_{M,w}$ has $3l$ input gates, where $l=2m+k$.

\item If $\vec{a}=(a_1,\ldots,a_l)\in\{0,1\}^l$ is a tuple
which encodes a configuration $\alpha$
such that $\gamma(s^\alpha)=\acc$,
then $C_{M,w}$ accepts the input $\vec{a}\,\vec{a}\,\vec{a}$.

\item If $\vec{a}=(a_1,\ldots,a_l)$, $\vec{b}=(b_1,\ldots,b_l)$ 
and $\vec{c}=(c_1,\ldots,c_l)$ are tuples
in $\{0,1\}^l$ which encode configurations $\alpha$, $\beta_1$ and $\beta_2$
such that $\beta_1\not=\beta_2$, 
$\gamma(s^\alpha)=\forall$, 
$\alpha\mapsto_M\beta_1$ and $\alpha\mapsto_M\beta_2$,
then $C_{M,w}$ accepts the input $\vec{a}\,\vec{b}\,\vec{c}$.

\item If $\vec{a}=(a_1,\ldots,a_l)$ and $\vec{b}=(b_1,\ldots,b_l)$ 
are tuples
in $\{0,1\}^l$ which encode configurations $\alpha$ and $\beta$
such that $\gamma(s^\alpha)=\exists$ and $\alpha\mapsto_M\beta$,
then $C_{M,w}$ accepts the input $\vec{a}\,\vec{b}\,\vec{b}$.

\item If $\vec{a}=(1,\ldots,1)\in\{0,1\}^l$ and $\vec{b}=(b_1,\ldots,b_l)\in\{0,1\}^l$
is a tuple that encodes the initial configuration $\alpha_w$ of $M$ with input word $w$, 
then $C_{M,w}$ accepts the input $\vec{a}\,\vec{b}\,\vec{b}$.

\item $C_{M,w}$ does not accept any other input tuples $(a_1,\ldots,a_{3l})\in \{0,1\}^{3l}$.
\end{enumerate}
Clearly the conditions 1-6 above can be checked in polynomial time 
with respect to $l$, and accordingly with respect to the length $n$ of the input $w$.  
As a result, the circuit $C_{M,w}$ can be constructed 
in polynomial time from the input word $w$.

Assume first that $M$ accepts the input $w$. Then the initial configuration $\alpha_w$
of $M$ with input $w$ is in the set $\AC(M)$. 
Consider now the structure $\mA_C=(A_C,S_C)$ defined by the circuit $C:=C_{M,w}$. 
Let $P_0\subseteq A_C$ be the set of all $i=\num(a_1,\ldots,a_{l})$ such that
$(a_1,\ldots,a_{l})$ encodes a configuration $\alpha\in\AC(M)$. Using conditions
2-4 and the definition of $\AC(M)$ it is easy show that $P_0$ is $S_C$-persistent.
But then, by condition 5, $P=P_0\cup\{\num(1,\ldots,1)\}$ is an $S_C$-persistent
set such that $2^l\in P$, and consequently $C$ is a positive instance of $\SUPER$.

Assume then that $C:=C_{M,w}$ is a positive instance of $\SUPER$.
Then there is an $S_{C}$-persistent set $P$ such that $2^{l}=\num(1,\ldots,1)\in P$. 
Let $P^M$ be the set of all configurations
$\alpha$ of $M$ such that $\num(a_1,\ldots,a_l)\in P$ for the tuple $(a_1,\ldots,a_l)$
that encodes $\alpha$.
By conditions 5 and 6, the initial configuration $\alpha_w$ is in $P^M$. 
Thus, it suffices to show that $P^M\subseteq\AC(M)$. 

Suppose this is not the case, i.e., $P^M\setminus\AC(M)\not=\emptyset$.
Since $P^M$ is finite, and $\mapsto_M$ is acyclic, then there exists a
configuration $\alpha\in P^M\setminus\AC(M)$ which does not have 
$\mapsto_M$-successors in $P^M\setminus\AC(M)$. We divide the
argument into cases according to the type $\gamma(s^\alpha)$ of the 
state of $\alpha$.
\begin{itemize}

\item Observe first that $\gamma(s^\alpha)=\acc$ is not possible, since 
$\alpha\not\in\AC(M)$. 

\item Assume that $\gamma(s^\alpha)=\rej$. Let $(a_1,\ldots,a_l)\in \{0,1\}^l$ be the
tuple that encodes $\alpha$. Then by conditions 2-6, there are no tuples
$(b_1,\ldots,b_l),(c_1,\ldots,c_l)$ such that 
$(\num(a_1,\ldots,a_l),\num(b_1,\ldots,b_l),\num(c_1,\ldots,c_l))\in S_C$.
This means that $\alpha\not\in P^M$, contrary to our assumption.

\item If $\gamma(s^\alpha)=\forall$, then by conditions 3 and 6 we see that 
$\beta_1,\beta_2\in P^M$, where $\beta_1$ and $\beta_2$ are the 
$\mapsto_M$-successors of $\alpha$. Since $\alpha$ has no $\mapsto_M$-successors 
in $P^M\setminus\AC(M)$, we have $\beta_1,\beta_2\in\AC(M)$. But then
by the definition of $\AC(M)$, also $\alpha\in\AC(M)$, contrary
to our assumption. 

\item If $\gamma(s^\alpha)=\exists$,
then by conditions 4 and 6, $\beta_1\in P^M$ or $\beta_2\in P^M$, 
where $\beta_1$ and $\beta_2$ are the 
$\mapsto_M$-successors of $\alpha$.
Since $\alpha$ has no $\mapsto_M$-successors 
in $P^M\setminus\AC(M)$, it follows that either $\beta_1\in\AC(M)$ or 
$\beta_2\in\AC(M)$. As a result, by the definition of $\AC(M)$, we have
$\alpha\in\AC(M)$, contrary to our assumption.

\end{itemize}
Since all the cases lead to contradiction, we conclude that $P^M\subseteq \AC(M)$.
\qed
\end{proof}

We will next show that $\SUPER$ is polynomial time reducible 
to the satisfiability problem of $\MINC$ with lax semantics, and in view of this
the latter is also $\EXPTIME$-hard.
In the proof we use the following notation: If $T$ is a team 
and $p_1,\ldots, p_n$ are proposition symbols, then $T(p_1,\ldots,p_n)$
is the set of all tuples $(a_1,\ldots,a_n)\in\{0,1\}^n$ such that for some $w\in T$,
$
	a_t=1\iff w\in V(p_t)\text{ for }t\in\{1,\ldots,n\}.
$
Note that the semantics of inclusion atoms can now be expressed as
$$
	M,T\models p_1\cdots p_n \subseteq q_1\cdots q_n \iff T(p_1,\ldots,p_n)\subseteq T(q_1,\ldots,q_n).
$$

\begin{theorem}\label{laxlower}
The satisfiability and finite satisfiability problems for \Minc with lax semantics
are hard for \EXPTIME with respect to $\leqpm$-reductions.
\end{theorem}
\begin{proof}
Let $C$ be a Boolean circuit with $3l$ input gates. 
Let $g_1,\ldots,g_m$ be the gates of $C$, where $g_1,\ldots, g_{3l}$
are the input gates and $g_m$ is the output gate.
We fix a distinct Boolean variable $p_i$ for each gate $g_i$. 
Let $\Phi$ be the set $\{p_1,\ldots,p_m\}$ of proposition symbols.
We define for each $i\in\{3l+1,\ldots,m\}$ a formula $\theta_i\in\MINC(\Phi)$ that describes 
the correct operation of the gate $g_i$:
$$
	\theta_i=\begin{cases}
	p_i\leftrightarrow \lnot p_j& \text{if $g_i$ is a NOT gate
	with input $g_j$}\\
	p_i\leftrightarrow (p_j\land p_k)& \text{if $g_i$ is an AND gate
	with inputs $g_j$ and $g_k$}\\
	p_i\leftrightarrow (p_j\lor p_k) &\text{if $g_i$ is an OR gate
	with inputs $g_j$ and $g_k$}
	\end{cases}
$$
Let $\psi_C$ be the formula $\bigl( \bigwedge_{3l+1\le i\le m}\theta_i \bigr)\;\land \, p_m$.
Thus, $\psi_C$ essentially says that the truth values of $p_i$, $1\le i\le m$, match an  
accepting computation of $C$.

Now we can define a formula $\varphi_C$ of $\MINC(\Phi)$ which is satisfiable
if and only if $C$ is a positive instance of $\SUPER$.
For the sake of readability, we denote here the variables corresponding 
to the input gates $g_{l+1},\ldots,g_{2l}$ by $q_1,\ldots,q_l$. Similarly,
we denote the variables $p_{2l+1},\ldots,p_{3l}$ by $r_1,\ldots,r_l$.
$$\label{pg:varphiC}
	\varphi_C:= \psi_C
	\land\; q_1\cdots q_l\subseteq p_1\cdots p_l \;
	\land\; r_1\cdots r_l\subseteq p_1\cdots p_l\;
	\land\; p_m\cdots p_m \subseteq p_1\cdots p_l.
$$
Note that $\varphi_C$ can clearly be constructed from the circuit $C$ in polynomial
time.

Assume first that $\varphi_C$ is satisfiable. That being so there is a Kripke model
$M=(W,R,V)$ and a nonempty team $T$ of $M$ such that $M,T\modelslax\varphi_C$.
Consider the model $\mA_C=(A_C,S_C)$ that corresponds to the circuit $C$.
We define a subset $P$ of $A_C$ as follows: $P:=\{\num(a_1,\ldots,a_l)\mid (a_1,\ldots,a_l)\in T(p_1,\ldots,p_l)\}.$

Observe first that since $M,T\modelslax p_m$ and $M,T\modelslax p_m\cdots p_m \subseteq p_1\cdots  p_l$,
$(1,\ldots,1)\in T(p_1,\ldots,p_l)$ and that being so $2^l=\num(1,\ldots,1)\in P$.
On that account, it suffices to show that  $P$ is $S_C$-persistent. To prove this,
assume that $i=\num(a_1,\ldots, a_l)\in P$. Then there is
a state $w\in T$ such that $w\in V(p_t) \iff a_t=1\text{\quad for }1\le t\le l.$

Define now $b_t,c_t\in \{0,1\}$, $1\le t\le l$, by the condition
$$
	b_t=1\iff w\in V(q_t)\text{\quad and \quad} 
	c_t=1\iff w\in V(r_t).
$$

As $M,T\modelslax \psi_C$, it follows from flatness (see Proposition \ref{flatness}) that $M,w\Vdash\psi_C$. 
By the definition of $\psi_C$, this means that the circuit $C$ accepts the 
input tuple 
$
	(a_1,\ldots,a_l,b_1,\ldots,b_l,$ $c_1,\ldots,c_l).
$ 
That being the case, $(i,j,k)\in S_C$, where $j=\num(b_1,\ldots,b_l)$ and $k=\num(c_1,\ldots,c_l)$. 

We still need to show that $j,k\in P$.  To see this, note that since
$M,T\modelslax q_1\cdots q_l\subseteq p_1\cdots p_l$, there exists
$w'\in T$ such that
$$
	w'\in V(p_t)\iff w\in V(q_t) \iff b_t=1\text{\quad for }1\le t\le l.
$$
Accordingly, $(b_1,\ldots,b_l)\in T(p_1,\ldots,p_n)$, and on that account $j\in P$. Similarly
we see that $k\in P$.

To prove the other implication, assume that $C$ is a positive instance of the
problem $\SUPER$. Then there is an $S_C$-persistent set $P\subseteq A_C$
such that $2^l\in P$. We let $M=(W,R,V)$ be the Kripke model 
and $T$ the team of $M$ such that 
\begin{itemize}
\item $T=W$ is the set of all tuples $(a_1,\ldots,a_m)\in \{0,1\}^m$ that correspond 
to an accepting computation of $C$ and for which 
$\num(a_1,\ldots,a_l),$ $\num(a_{l+1},\ldots,a_{2l}),$ $\num(a_{2l+1},\ldots,a_{3l})\in P$,
\item $R=\emptyset$, and $V(p_t)=\{(a_1,\ldots,a_m)\in W\mid a_t=1\}$ for $1\le t\le m$.
\end{itemize}
We will now show that $M,T\modelslax\varphi_C$, and accordingly $\varphi_C$ is
satisfiable. Note first that $M,T\modelslax\psi_C$, since by the definition
of $T$ and $V$, for any $w\in T$, the truth values of $p_i$ in $w$ correspond
to an accepting computation of $C$.

To prove  $M,T\modelslax q_1\cdots q_l\subseteq p_1\cdots p_l$, assume
that $(b_1,\ldots,b_l)\in T(q_1,\ldots,q_l)$. Then $i:=\num(b_1,\ldots,b_l)\in P$,
and since $P$ is $S_C$-persistent, there are $j,k\in P$ such that
$(i,j,k)\in S_C$. Accordingly, there is a tuple $(a_1,\ldots, a_m)\in\{0,1\}^m$
corresponding to an accepting computation of $C$
such that $(a_1,\ldots,a_l)=(b_1,\ldots,b_l)$, $j=\num(a_{l+1},\ldots,a_{2l})$ and
$k=\num(a_{2l+1},\ldots,a_{3l})$. This means that $(a_1,\ldots,a_m)$ is in $T$,
and that being the case $(b_1,\ldots,b_l)\in T(p_1,\ldots,p_l)$. The claim
$M,T\modelslax r_1\cdots r_l\subseteq p_1\cdots p_l$ is proved in the same way.

Note that since $M,T\models p_m$, we have $T(p_m,\ldots,p_m)=\{(1,\ldots,1)\}$.
Furthermore, since $2^l=\num(1,\ldots,1)\in P$ and $P$ is $S_C$-persistent, 
there is an element $(a_1,\ldots,a_m)\in T$ such that  
$(a_1,\ldots,a_l)=(1,\ldots,1)$.
Consequently, we see that $(1,\ldots,1)\in T(p_1,\ldots, p_l)$, 
and consequently $M,T\modelslax p_m\cdots p_m \subseteq p_1\cdots p_l$.
\qed
\end{proof}

\begin{corollary}
The satisfiability and finite satisfiability problems of modal inclusion logic with lax semantics
are $\EXPTIME$-complete with respect to $\leqpm$-reductions.
\end{corollary}

Note that the formula $\varphi_C$ used in the proof of Theorem \ref{laxlower}
is in \emph{propositional inclusion logic}, i.e., it  does not contain any modal operators.
In view of this, our proof shows that the satisfiability problem of propositional inclusion logic
is already $\EXPTIME$-hard. Naturally, this problem is also in $\EXPTIME$,
since propositional inclusion logic is a fragment of $\MINC$.

\begin{corollary}
The satisfiability and finite satisfiability problems of propositional inclusion logic with lax semantics
are $\EXPTIME$-complete with respect to $\leqpm$-reductions.
\end{corollary}

\section{Complexity of Satisfiability for Strict Semantics}\label{strictall}


We now show that the satisfiability and finite 
satisfiability problems for $\Minc$ with strict semantics are in $\EXPTIME$.
The proof is a simple adaptation of the upper bound argument
for lax semantics from the proof of Theorem~\ref{thm-sat-fsat-lax-exptime}, but uses the logic $\mathcal{GC}^2$, that is,
\emph{two-variable guarded fragment with counting}.
Both, the standard and finite satisfiability problems of
this logic are $\EXPTIME$-complete, as
has been shown by Kazakov \cite{kazakov} and Pratt-Hartmann \cite{pratt-hart}, respectively.
The set of formulae of the logic $\mathcal{GC}^2$ is
the smallest set $S$ that satisfies the following conditions.
\begin{enumerate}
\item
The set $S$ contains all atomic relational formulae 
that use only the fixed variables $x$ and $y$.
Equalities are also allowed.
\item
The set $S$ is closed under the standard Boolean operators.
\item
If $\varphi(u)\in S$ has at most one free variable $u\in\{x,y\}$,
then the formulae $\exists u\varphi(u)$ and $\forall u\varphi(u)$ are in $S$.
\item
Let $\gamma$ denote a \emph{guard atom}, i.e, a binary relational atom  of
the type $Rxy$ or $Ryx$, where $R$ is any binary relation
symbol other than equality.  Let $Q$ denote
any of the quantifiers $\exists$, $\exists^{\leq k}$, $\exists^{\geq k}$, $\exists^{= k}$,
where $k$ denotes any positive integer (encoded in binary). Let $u\in\{x,y\}$ be a variable,
and let $\varphi$ be any formula in $S$.
Then the \emph{guarded formulae} $Qu(\gamma\wedge
\varphi)$ and $\forall u(\gamma\rightarrow\varphi)$ are in $S$.
\end{enumerate}
The semantics of $\mathcal{GC}^2$ is clear. We will now show how
the satisfiability and finite satisfiability problems of \Minc
with strict semantics are reduced to the
corresponding problems for $\mathcal{GC}^2$.

\begin{theorem}\label{thm:upperBoundStrict}
The satisfiability and finite satisfiability problems for $\Minc$ with strict semantics are in $\EXPTIME$.
\end{theorem}
\begin{proof}
Let $\theta$ be a formula of $\Minc$. An equisatisfiable translation
of $\theta$ is obtained
from the formula $\varphi_{\theta}$, which we
defined just before Theorem \ref{thm-sat-fsat-lax-exptime} in Section \ref{laxup}.
It is clear that $\varphi_{\theta}$ translates via a simple
extension of the \emph{standard translation}
into $\mathcal{GC}^2$; see \cite{blackburn} for the standard
translation of modal logic. Let $t(\varphi_{\theta})$ denote the $\mathcal{GC}^2$ -formula
obtained by using the (extension of the) standard translation.
For each $\varphi\in\mathit{SUB}(\varphi_{\theta})$,
let $t(\chi_{\varphi})$ denote the translation of 
the subformula $\chi_{\varphi}$ of $\varphi_{\theta}$;
see the argument for lax semantics for the definition of 
the formulas $\chi_{\varphi}$.
The only thing we now need to do is to modify 
the formulas $t(\chi_{\Diamond\varphi})$ 
and $t(\chi_{\varphi\vee\psi})$.
In the case of $t(\chi_{\varphi\vee\psi})$,
we simply add a conjunct stating 
that the unary predicates $p_{\varphi}$ and $p_{\psi}$
are interpreted as disjoint sets: $\neg\exists x( p_{\varphi}(x)\wedge p_{\psi}(x))$.
To modify the formulas $t(\chi_{\Diamond\varphi})$, 
we appoint a novel binary  relation $R_{\Diamond\varphi}$ for 
each formula $\Diamond\varphi\in\mathit{SUB}(\theta)$.
%
%
%
%
%
%
%
%
%
%
%
%
%
%
%
%
%
%
We then define the formula $\beta$ which  states that $R_{\Diamond\varphi}$
is a function from the interpretation of $p_{\Diamond\varphi}$ onto
the interpretation of $p_{\varphi}$.
\begin{multline*}
\beta\dfn \forall{x}\bigl( p_{\Diamond\varphi}(x)
\rightarrow \exists^{=1} y (R_{\Diamond\varphi}xy \wedge p_{\varphi}(y)\bigr)
\wedge \forall x\forall y\bigl(
R_{\Diamond\varphi}xy\rightarrow (p_{\Diamond\varphi}(x)\wedge p_{\varphi}(y))\bigr)\\
\wedge \forall y\bigl( p_{\varphi}(y)\rightarrow \exists x(p_{\Diamond\varphi}(x)
\wedge R_{\Diamond\varphi}xy)\bigr).
\end{multline*}
Define $\beta' \dfn \forall x\forall y\bigl(R_{\Diamond\varphi}xy
\rightarrow Rxy\bigr)$, where $R$ is the accessibility relation of
modal inclusion logic. The conjunction $\beta\wedge\beta'$ is the
desired  modification of 
$t(\chi_{\Diamond\varphi})$.
The modification of $t(\varphi_\theta)$,
using the modified versions of 
$t(\chi_{\varphi\vee\psi})$
and $t(\chi_{\Diamond\varphi})$,
is the desired $\mathcal{GC}^2$ -formula equisatisfiable with $\theta$.
The remaining part of the proof is practically identical to the corresponding
argument for lax semantics.
\end{proof}

Finally, we will turn to prove the corresponding lower bound in order to achieve the desired completeness result.

\begin{theorem}
The satisfiability and finite satisfiability problems for $\Minc$ with strict semantics are $\EXPTIME$-hard under $\leqpm$-reductions.	
\end{theorem}
\begin{proof}
The proof of Theorem~\ref{laxlower} works without changes for strict semantics. 
This is evident from the MINC-formula $\varphi_C$ (p.~\pageref{pg:varphiC}) used in the proof: 
It is purely propositional, so the difference between lax and strict semantics for the diamond operator is irrelevant. 
Furthermore, disjunctions only occur in the conjunct $\psi_C$ which is flat as it does not contain any inclusion atoms.
Consequently, the difference between lax and strict semantics for disjunction does not have any effect.
\end{proof}

\begin{corollary}
	The satisfiability and finite satisfiability problems for $\Minc$ with strict semantics are $\EXPTIME$-complete under $\leqpm$-reductions.
\end{corollary}

\section{Discussion}\label{trees}
The conference version of this paper \cite{hkmv15} claims different complexities regarding the satisfiability problem with respect to the underlying semantics.
The satisfiability problem under lax semantics is there shown to be $\EXPTIME$-complete and
claimed $\NEXPTIME$-complete for strict semantics.
The proof for strict semantics is incorrect. The proof argues by enforcing assignment trees through formula gadgets similarly as in Ladner's proof for satisfiability in modal logic \cite{lad77}.
The idea does work, however, on \emph{specific structures}, as we will show later in this section.

Let us be more detailed in the following.
The main procedure was the following.
We aimed to reduce the \NEXPTIME-complete dependency quantifier version of QBF (DQBF) to a variant of QBF with inclusion atoms and eventually to the satisfiability in $\MINC$.
Accordingly, we needed to express the QBF results in modal logic with strict semantics. 
In particular, propositional dependence atoms must be translated into propositional inclusion atoms. 
This is possible if we can simulate ``strict quantification over Boolean values'' by diamonds with strict semantics. 
The main idea was to force models to be of the structure of assignment trees by using a well-known technique from Ladner \cite{lad77}.
In theory, this approach works as long as we stay in a single assignment tree.
 
However, in general, the following problem arises:
Consider a single model consisting of two isomorphic assignment trees ($A$ and $B$) and a team that consists of exactly the roots of the two assignment trees. 
We begin evaluating our modal inclusion logic formula from there. 
Let $f$ be an isomorphism from tree $A$ onto tree $B$.
Consider some node $w$ in $A$. 
Now, the strict diamond sends a node $w$ in $A$ to \emph{exactly one} witness successor $u$ of $w$. 
Similarly, the strict diamond sends node $f(w)$ to exactly one witness successor $u'$ of $f(w)$. 
However, it may happen that $u$ is the ``left'' successor of $w$ and $u'$ is the ``right'' successor of $f(w)$, i.e., the nodes $f(u)$ and $u'$ are two separate nodes.
Intuitively, we have now chosen \emph{two} values for a proposition symbol $p$, ``true'' and ``false'': 
one value is realised in $u$ and the other one in $u'$. 
Consequently, intuitively and informally, we are using lax semantics by choosing two values that extend the propositional valuation function associated with the path that goes from the root of $A$ to $w$; one of these values is actually in the tree $B$, but that will not save us.

As a result, we fail to simulate strict Boolean quantification with the strict diamond, and accordingly, the definition of propositional dependence atoms using propositional inclusion atoms fails, as that definition would require that our ``Boolean quantification'' is strict. 
In a single assignment tree everything would work, but not in two.
It may be illustrative to consider the formula $\Diamond(p\subseteq\neg p)$. 
Under strict semantics this formula is not satisfiable by any singleton team. 
However, it is easy to satisfy this formula (still under strict semantics) in a team $\{x,y\}$, where both $x$ and $y$ are roots of two isomorphic trees.

In the previous section, we have seen how the result can be corrected and shown to
have the same complexity as lax semantics.
Now we will show that the proof idea of the faulty result works in a specific class of models, in particular, structures whose underlying graph is a binary tree.
In the following we introduce the necessary definitions and eventually present the mentioned result (split into the upper and lower bound cases).
In consequence, we achieve a specific $\NEXPTIME$-completeness result for satisfiability with \emph{strict} semantics.
\bigskip

We say a Kripke model $M=(W,R,V)$ is a \emph{binary tree model} if $(W,R)$ is a (possibly infinite) directed tree where every node has out-degree two or zero.

Now we show that the satisfiability and finite satisfiability problems for $\Minc$ with strict semantics restricted to binary trees are in $\NEXPTIME$. 
More precisely, let $\BinTreeMincstrictSAT$ ($\FinBinTreeMincstrictSAT$) be the following problem: 
given a formula $\varphi\in\Minc$, does there exist a (finite) binary tree model $M$ such that $M, \{r\}\models\varphi$, where $r$ is the root of $M$.

\begin{theorem}\label{thm:upperBoundStrictTrees}
$\BinTreeMincstrictSAT$ and $\FinBinTreeMincstrictSAT$ are in \NEXPTIME. 
\end{theorem}
\begin{proof}
In both cases, general and finite satisfiability, we first guess a Kripke model
which is a binary tree whose depth is at most
the modal depth of the input formula. Then we model-check the formula in the 
guessed model; the model-checking problem for modal inclusion logic under
strict semantics is NP-complete \cite{modelchminc}.
\end{proof}

In the remainder of this section we will show a matching lower bound via a chain of reductions.
These reductions use quantified variants of dependence- and inclusion logics.

\begin{theorem}\label{thm:sat-trees-nexptime-hard}
$\BinTreeMincstrictSAT$ and $\FinBinTreeMincstrictSAT$ are \NEXPTIME-hard under $\leqlogm$-reductions.
\end{theorem}

To prove this result, we will show how to reduce from a dependence variant of QBF validity to an inclusion variant of QBF validity, and finally to satisfiability of \Minc with strict semantics restricted to binary trees.

\newcommand{\QDL}{\mathsf{QDPL}}
The notion we will define shortly, is following Hannula \cite{h16} (also
discussed in \cite{hkmv15}). The set of formulas of \emph{quantifier propositional dependence logic} $\QDL$ is defined inductively by the following grammar:
$$
\varphi \ddfn
			p\mid 
			\lnot p\mid 
			(\varphi_1\land\varphi_2)\mid 
			(\varphi_1\lor\varphi_2)\mid
			\dep{p_1,\dots,p_k,q}\mid
			\forall p\, \varphi\mid
			\exists p\, \varphi,
$$
where $p,p_1,\dots,p_k,q$ are propositions and $k\in\mathbb N$. 
The semantics of $\QDL$ is then defined as follows, where $T$ is a set of assignments.
$$
\begin{array}{lll}
	T\models \forall p\,\psi& \text{ iff }& \{s(a/p) \mid s \in T, a \in \{0, 1\}\}\models \psi,\\
	T\models \exists p\, \psi & \text{ iff }& \exists F\in {}^T\{\{0\},\{1\}\}: \{s(a/p) \mid s \in T,a \in F(s)\}\models \psi,\\
	T\models\dep{p_1,\dots,p_k,q} & \text{ iff } & \forall s,s'\in T:s(\vec p)=s'(\vec p)\text{ implies }s(q)=s'(q)
\end{array}
$$
Observe that the existential quantifier is defined via strict semantics, i.e., the tuple $\{0,1\}$ is missing in the expression on the right hand side of the definition.
The connectives $\vee$ and $\wedge$ are interpreted exactly as in the case of modal inclusion logic using strict semantics.
Literals $p$, $\neg p$ are also interpreted as in modal inclusion logic.

It is straightforward to show that $\QDL$ is \emph{downwards closed}: 
for any formula $\varphi$ of $\QDL$, if $T\models\varphi$ and
$T'\subseteq T$, then $T'\models\varphi$. (See \cite[Proposition 3.10]{vaananen07} for a proof
in the case of first-order dependence logic.)

The syntax and semantics of \emph{quantifier propositional inclusion logic}, $\IQBF$, is defined in the same way as $\QDL$, except that dependence atoms are replaced by
inclusion atoms $\vec p\subseteq \vec q$. In particular, we use strict semantics for existential quantification.
The semantics of $\vec p\subseteq \vec q$ is given by the condition $T\models\vec p\subseteq\vec q$ if and only if $\forall s\in T\exists s'\in T: s(\vec p)=s'(\vec q)$. 

We denote the validity problem for sentences of $\QDL$ by $\QDL\text-\VAL$. 
Similarly, we denote the validity problem for sentences of $\IQBF$ by $\IQBFval$. 

\begin{proposition}[{\cite{h16}}]\label{prop:qpdlval-nexptime}
	$\QDL\text-\VAL$ is $\NEXPTIME$-complete under $\leqpm$-reductions.
\end{proposition}

Note that the preceding result was originally shown in
\cite{h16} for \emph{lax} semantics. However, the result for strict semantics follows easily by downwards closure.

\begin{lemma}\label{lem:qpdlval-qplinc-val}
	$\QDL\text-\VAL\leqlogm\IQBFval$. 
\end{lemma}
\begin{proof}
We translate the expressions $\mathrm{dep}(p_1,\dots,p_k,q)$ to inclusion atoms in a way to be described next.
Inspired by \citeauthor{ghk13} \cite[Corollary 23]{ghk13}, we observe that,
under strict semantics, inclusion atoms can simulate formulas $\mathrm{dep}(p_1,\dots,p_k,q)$, as the following example demonstrates:
$$\forall p\forall q\exists r(\dep{q,r}\land\varphi)\text{ is equivalent to }\forall p\forall q\exists r(\forall s(sqr\subseteq pqr)\land\varphi),$$ 
where $\varphi$ is a quantifier-free formula with free variables among $\{p,q,r\}$. 
Any sentence $\psi$ of $\QDL$ is equivalent to a sentence of the form 
$$
\forall \vec q \,\exists \vec r\left(\bigwedge_{r_i\in\vec r}\dep{\vec{q_i},r_i}\land\varphi\right),
$$
where for all $i$, the variables in $\vec{q_i}$ are contained in $\vec q$, $\vec r$ and $\vec q$ are disjoint, and $\varphi$ is a quantifier-free formula. This normal form is proved by \citeauthor{ghk13} \cite[Corollary 23]{ghk13} for \emph{first-order dependence logic}; since $\QDL$ is essentially the restriction of first-order dependence logic to Boolean 
structures, the normal form also holds for $\QDL$.

Simulating each dependence atom in the normal form of $\psi$ with an
inclusion atom, as in the example above, we see that $\psi$ is equivalent 
to
$$
\forall \vec q \,\exists \vec r\left(\bigwedge_{r_i\in\vec r}\forall \vec{s_i}(\vec{s_i}\vec{q_i}r_i\subseteq\vec{p_i}\vec{q_i}r_i)\land\varphi\right),
$$
where for all $i$, $\vec{p_i}$ contains those variables in $\vec q$
that are not in $\vec{q_i}$ and $\vec{s_i}$ is a fresh tuple of variables of the same length as $\vec{p_i}$. Thus, there is a validity preserving
translation from $\QDL$ to $\IQBF$.
\qed	
\end{proof}

\noindent Now, for the last step, we explain how \IQBFval finally reduces to $\FinBinTreeMincstrictSAT$. 
\begin{lemma}\label{lem:qplinc-mincstrictsat-trees}
	$\IQBFval\leqlogm\BinTreeMincstrictSAT$ and 
	$\IQBFval\leqlogm\FinBinTreeMincstrictSAT$.
\end{lemma}
\begin{proof}
This proof is just a slight modification of the standard proof by Ladner showing $\PSPACE$-hardness of plain modal logic via a reduction from QBF validity \cite{lad77}. 
The idea is to enforce a complete assignment tree, and as we are restricted to binary trees it is only required to map the variables in the correct way to the nodes of the tree(s). 
Further, one uses clause propositions which are true if the corresponding literal holds. 
Following the work of Ladner \cite{lad77}, define the formula which enforces the described substructure by $\varphi_{\text{struc}}$ as follows. Let $r_1,\dots, r_n$ be the variables of the given $\IQBFval$ instance $\Game_1 r_{1}\Game_2 r_{2}\cdots\Game_n r_{n}(\varphi\land\chi)$ where $\varphi$ is the conjunctive normal form formula and $\chi$ is the conjunction of the inclusion atoms (stemming from the translation in the proof of Lemma~\ref{lem:qpdlval-qplinc-val}), then
 \newcommand{\branch}[1]{\protect\ensuremath{\textit{branch}(#1)}}
 \newcommand{\store}[1]{\protect\ensuremath{\textit{store}(#1)}}
\begin{align*}
	\varphi_{\text{struc}}\dfn\branch{r_1}\land\bigwedge_{i=1}^{n-1}\Box^i\left(\branch{r_{i+1}}\land\bigwedge_{j=1}^i\store{r_j}\right),
\end{align*}

where
\begin{align*}
	\branch{r_i} &\dfn \Diamond r_i\land\Diamond\lnot r_i\\
	\store{r_i} &\dfn (r_i\to\Box r_i)\land(\lnot r_i\to\Box\lnot r_i).
\end{align*}

The final formula is then a formula of type $\varphi_{\text{struc}}\land\triangle_1\triangle_2\cdots\triangle_n(\varphi\land\chi)$, where $\triangle_i=\Box$ if $\Game_i=\forall$ and $\triangle_i=\Diamond$ if $\Game_i=\exists$. 
Let us denote this translation by the function $f$ which can be computed in polynomial time. 
Then it holds that $\varphi\in\IQBFval$ if and only if $f(\varphi)\in\BinTreeMincstrictSAT$. 
Clearly, this covers also the case for finite satisfiability. 
 \qed	
\end{proof}

\begin{proof}[of Theorem~\ref{thm:sat-trees-nexptime-hard}]
This result follows from Proposition~\ref{prop:qpdlval-nexptime} together with Lemmas~\ref{lem:qpdlval-qplinc-val} and \ref{lem:qplinc-mincstrictsat-trees}.\qed
\end{proof}

\noindent The following corollary follows from Theorem~\ref{thm:upperBoundStrictTrees} and Theorem~\ref{thm:sat-trees-nexptime-hard}.

\begin{corollary}
$\BinTreeMincstrictSAT$ and $\FinBinTreeMincstrictSAT$ are both $\NEXPTIME$-complete under $\leqpm$-reductions.
\end{corollary}

\section{Conclusion}\label{conclusion}
We have compared the strict and lax variants of team semantics from the perspective of satisfiability problems for modal inclusion logic \Minc. 
Interestingly, the problems do not differ in their complexities. 
Both lead to completeness for the class \EXPTIME under strong $\leqpm$-reductions. 
We have seen how the restrictions of the structures to binary trees allows for an increase in complexity, that is, for strict semantics the problem becomes $\NEXPTIME$-complete.
For further research it is left open to classify the complexity of the lax version under these specific type of structures.
\bigskip

\noindent\textbf{Acknowledgements.}
The second author acknowledges support from the ERC project 647289 ``CODA" and
Jenny and Antti Wihuri Foundation.
The third author is supported by DFG grant ME 4279/1-1.

\bibliographystyle{plainnat}
\bibliography{minc}

\begin{thebibliography}{18}
\providecommand{\natexlab}[1]{#1}
\providecommand{\url}[1]{\texttt{#1}}
\expandafter\ifx\csname urlstyle\endcsname\relax
  \providecommand{\doi}[1]{doi: #1}\else
  \providecommand{\doi}{doi: \begingroup \urlstyle{rm}\Url}\fi

\bibitem[Blackburn et~al.(2001)Blackburn, Rijke, and Venema]{blackburn}
Patrick Blackburn, Maarten~de Rijke, and Yde Venema.
\newblock \emph{Modal Logic}.
\newblock Cambridge Tracts in Theoretical Computer Science. Cambridge
  University Press, 2001.
\newblock \doi{10.1017/CBO9781107050884}.

\bibitem[Galliani(2012)]{Galliani12}
Pietro Galliani.
\newblock Inclusion and exclusion dependencies in team semantics - on some
  logics of imperfect information.
\newblock \emph{Ann. Pure Appl. Logic}, 163\penalty0 (1):\penalty0 68--84,
  2012.
\newblock \doi{10.1016/j.apal.2011.08.005}.
\newblock URL \url{https://doi.org/10.1016/j.apal.2011.08.005}.

\bibitem[Galliani et~al.(2013)Galliani, Hannula, and Kontinen]{ghk13}
Pietro Galliani, Miika Hannula, and Juha Kontinen.
\newblock Hierarchies in independence logic.
\newblock In Simona Ronchi~Della Rocca, editor, \emph{Computer Science Logic
  2013 {(CSL} 2013), {CSL} 2013, September 2-5, 2013, Torino, Italy}, volume~23
  of \emph{LIPIcs}, pages 263--280. Schloss Dagstuhl - Leibniz-Zentrum fuer
  Informatik, 2013.
\newblock ISBN 978-3-939897-60-6.
\newblock \doi{10.4230/LIPIcs.CSL.2013.263}.
\newblock URL
  \url{http://drops.dagstuhl.de/opus/portals/extern/index.php?semnr=13009}.

\bibitem[Greenlaw et~al.(1995)Greenlaw, Hoover, and Ruzzo]{greenlaw}
Raymond Greenlaw, H.~James Hoover, and Walter~L. Ruzzo.
\newblock \emph{Limits to Parallel Computation: P-completeness Theory}.
\newblock Oxford University Press, Inc., New York, NY, USA, 1995.
\newblock ISBN 0-19-508591-4.

\bibitem[Hannula(2016)]{h16}
Miika Hannula.
\newblock The entailment problem in modal and propositional dependence logics.
\newblock \emph{CoRR}, abs/1608.04301, 2016.
\newblock URL \url{http://arxiv.org/abs/1608.04301}.

\bibitem[Hella et~al.(2015)Hella, Kuusisto, Meier, and Vollmer]{hkmv15}
Lauri Hella, Antti Kuusisto, Arne Meier, and Heribert Vollmer.
\newblock Modal inclusion logic: Being lax is simpler than being strict.
\newblock In Giuseppe~F. Italiano, Giovanni Pighizzini, and Donald Sannella,
  editors, \emph{Mathematical Foundations of Computer Science 2015 - 40th
  International Symposium, {MFCS} 2015, Milan, Italy, August 24-28, 2015,
  Proceedings, Part {I}}, volume 9234 of \emph{Lecture Notes in Computer
  Science}, pages 281--292. Springer, 2015.
\newblock ISBN 978-3-662-48056-4.
\newblock \doi{10.1007/978-3-662-48057-1_22}.

\bibitem[Hella et~al.(2016)Hella, Kuusisto, Meier, and
  Virtema]{modelchmincCorr}
Lauri Hella, Antti Kuusisto, Arne Meier, and Jonni Virtema.
\newblock Model checking and validity in propositional and modal inclusion
  logics.
\newblock \emph{CoRR}, abs/1609.06951, 2016.

\bibitem[Hella et~al.(2017)Hella, Kuusisto, Meier, and Virtema]{modelchminc}
Lauri Hella, Antti Kuusisto, Arne Meier, and Jonni Virtema.
\newblock Model checking and validity in propositional and modal inclusion
  logics.
\newblock In \emph{Proc. MFCS}, 2017.
\newblock Full version in CoRR \cite{modelchmincCorr}.

\bibitem[Hemaspaandra(1996)]{hema}
Edith Hemaspaandra.
\newblock The price of universality.
\newblock \emph{Notre Dame Journal of Formal Logic}, 37\penalty0 (2):\penalty0
  174--203, 1996.
\newblock \doi{10.1305/ndjfl/1040046086}.

\bibitem[Kazakov(2004)]{kazakov}
Yevgeny Kazakov.
\newblock A polynomial translation from the two-variable guarded fragment with
  number restrictions to the guarded fragment.
\newblock In Jos{\'{e}}~J{\'{u}}lio Alferes and Jo{\~{a}}o~Alexandre Leite,
  editors, \emph{Logics in Artificial Intelligence, 9th European Conference,
  {JELIA} 2004, Lisbon, Portugal, September 27-30, 2004, Proceedings}, volume
  3229 of \emph{Lecture Notes in Computer Science}, pages 372--384. Springer,
  2004.
\newblock ISBN 3-540-23242-7.
\newblock \doi{10.1007/978-3-540-30227-8_32}.

\bibitem[Kontinen et~al.(2017)Kontinen, M{\"{u}}ller, Schnoor, and
  Vollmer]{KMSV14}
Juha Kontinen, Julian{-}Steffen M{\"{u}}ller, Henning Schnoor, and Heribert
  Vollmer.
\newblock Modal independence logic.
\newblock \emph{J. Log. Comput.}, 27\penalty0 (5):\penalty0 1333--1352, 2017.
\newblock \doi{10.1093/logcom/exw019}.

\bibitem[Kuusisto(2015)]{Ku15}
Antti Kuusisto.
\newblock A double team semantics for generalized quantifiers.
\newblock \emph{Journal of Logic, Language and Information}, 24\penalty0
  (2):\penalty0 149--191, 2015.
\newblock \doi{10.1007/s10849-015-9217-4}.

\bibitem[Ladner(1977)]{lad77}
Richard~E. Ladner.
\newblock The computational complexity of provability in systems of modal
  propositional logic.
\newblock \emph{{SIAM} J. Comput.}, 6\penalty0 (3):\penalty0 467--480, 1977.
\newblock \doi{10.1137/0206033}.

\bibitem[Pratt{-}Hartmann(2007)]{pratt-hart}
Ian Pratt{-}Hartmann.
\newblock Complexity of the guarded two-variable fragment with counting
  quantifiers.
\newblock \emph{J. Log. Comput.}, 17\penalty0 (1):\penalty0 133--155, 2007.
\newblock \doi{10.1093/logcom/exl034}.
\newblock URL \url{https://doi.org/10.1093/logcom/exl034}.

\bibitem[Sevenster(2009)]{sev09}
Merlijn Sevenster.
\newblock Model-theoretic and computational properties of modal dependence
  logic.
\newblock \emph{J. Log. Comput.}, 19\penalty0 (6):\penalty0 1157--1173, 2009.
\newblock \doi{10.1093/logcom/exn102}.

\bibitem[V\"a\"an\"anen(2008)]{vaananen08b}
Jouko V\"a\"an\"anen.
\newblock {M}odal {D}ependence {L}ogic.
\newblock In K.~Apt and R.~van Rooij, editors, \emph{New Perspectives on Games
  and Interaction}, pages 237--254. Amsterdam University Press, 2008.

\bibitem[V{\"{a}}{\"{a}}n{\"{a}}nen(2007)]{vaananen07}
Jouko~A. V{\"{a}}{\"{a}}n{\"{a}}nen.
\newblock \emph{Dependence Logic - {A} New Approach to Independence Friendly
  Logic}, volume~70 of \emph{London Mathematical Society student texts}.
\newblock Cambridge University Press, 2007.
\newblock ISBN 978-0-521-70015-3.
\newblock URL
  \url{http://www.cambridge.org/de/knowledge/isbn/item1164246/?site_locale=de_DE}.

\bibitem[van Eijck(2014)]{eijck}
Jan van Eijck.
\newblock Dynamic epistemic logics.
\newblock In Alexandru Baltag and Sonja Smets, editors, \emph{Johan van Benthem
  on Logic and Information Dynamics}, pages 175--202. Springer, 2014.
\newblock ISBN 978-3-319-06024-8.
\newblock \doi{10.1007/978-3-319-06025-5_7}.

\end{thebibliography}
\label{LastPage}
\end{document}